\documentclass[conference]{IEEEtran}
\IEEEoverridecommandlockouts

\usepackage{cite}
\usepackage{algorithmic}
\usepackage{graphicx}
\usepackage{textcomp}
\usepackage{xcolor}
\def\BibTeX{{\rm B\kern-.05em{\sc i\kern-.025em b}\kern-.08em
    T\kern-.1667em\lower.7ex\hbox{E}\kern-.125emX}}
\usepackage{amsmath}
\usepackage{amsthm}   
\usepackage{amssymb}
\usepackage{amsfonts}
\usepackage{mathtools}
\usepackage{enumitem}
\usepackage{algorithm}
\usepackage[algo2e]{algorithm2e}
\usepackage{tikz}
\usepackage{bm}
\usepackage{csvsimple}
\usepackage{booktabs}
\usepackage{siunitx}

\theoremstyle{plain} 
\newtheorem{theorem}{Theorem}[section]     
\newtheorem{corollary}[theorem]{Corollary} 
\newtheorem{lemma}[theorem]{Lemma}         
\newtheorem{proposition}[theorem]{Proposition} 

\theoremstyle{definition} 
\newtheorem{definition}[theorem]{Definition}

\newtheorem{remark}[theorem]{Remark}
\newtheorem{problem}[theorem]{Problem}
\newtheorem{assumption}{Assumption}
\begin{document}

\title{A Novel Discrete-Time Model of Information Diffusion on Social Networks Considering Users Behavior}

\author{\IEEEauthorblockN{1\textsuperscript{st} Tran Van Khanh}
	\IEEEauthorblockA{\textit{Faculty of Information Technology} \\
		\textit{Posts and Telecommunications}\\ \textit{Institute of Technology}\\
		Hanoi, Vietnam \\
		khanh080605@gmail.com}
	\and
	\IEEEauthorblockN{2\textsuperscript{nd} Do Xuan Cho}
	\IEEEauthorblockA{\textit{Faculty of Information Security} \\
		\textit{Posts and Telecommunications}\\ \textit{Institute of Technology}\\
		Hanoi, Vietnam \\
		chodx@ptit.edu.vn}
	\and
	\IEEEauthorblockN{3\textsuperscript{rd} Hoang Phi Dung}
	\IEEEauthorblockA{\textit{Faculty of Fundamental Sciences} \\
		\textit{Posts and Telecommunications}\\ \textit{Institute of Technology}\\
		Hanoi, Vietnam \\
		dunghp@ptit.edu.vn}

}

\maketitle

\begin{abstract}
In this paper, we introduce the SDIR (Susceptible–Delayable–Infected–Recovered) model, an extension of the classical SIR epidemic framework, to provide a more explicit characterization of user behavior in online social networks. The newly merged state D (delayable) represents users who have received the information but delayed its spreading and may eventually choose not to share it at all. Based on the mean-field approximation method, we derive the dynamical equations of the model and investigate its convergence and stability conditions. Under these conditions, we further propose a greedy algorithm and a sandwich approximation algorithm for the edge-deletion problem, aiming to minimize the influence of information diffusion by identifying approximate solutions.
\end{abstract}


\section{Introduction}
Online social networks have become one of the most crucial and essential information platforms for communication and commerce on a global scale. Because of their highly complex data structures, information spreading in social networks has emerged as an ideal environment for propagation—deep, wide, and significantly faster than any previous medium. Studying information diffusion on platforms such as Facebook, Twitter, and TikTok plays an essential role in communication media, information security, and the social sciences \cite{Chen2014, Kempe2015, Richardson2002, Shakarian2015}. In these platforms, the influence maximization (IM) problem \cite{Canh2018, Chen2014, Kempe2015, Shakarian2015} has been recognized as a fundamental problem in viral marketing, while the influence minimization (IMIN) problem \cite{Canh2019, Kimura2008, Shakarian2015, Xie2023} is central to controlling the spread of harmful or false information in online social networks.

Epidemic models have long been studied using mathematical formulations, for instance, the Susceptible-Infected-Recovered (SIR) model is used for epidemic forecasting in epidemiology \cite{Kermack1927}. Initially, epidemic models such as SIR and SIS (Susceptible-Infected-Susceptible) were applied to epidemiology for disease forecasting \cite{Bailey1975, Kermack1927}. Nowadays, however, these models are also widely applied to other fields, including viral advertising \cite{Phelps2004, Richardson2002}, cybersecurity \cite{Acemoglu2013, Alpcan2010}, and information diffusion \cite{Jacquet2010, Mieghem2009}. Models of information propagation are generally categorized into stochastic and deterministic approaches, with several comprehensive surveys available \cite{Nowzari2016, Pastor-Satorras2015}. For example, Wang et al. \cite{Wang2003} introduced a discrete-time virus spread model and evaluated the epidemic threshold as a function of network structure. Mieghem et al. \cite{Mieghem2009, Mieghem2014} applied Markov chain formulations to analyze biological diffusion models such as SIR, and subsequent works explored SIR variants, including SIS and SIRS \cite{Ahn2013, Ruhi2015}. Through mean-field approximation, stochastic SIR models can be transformed into deterministic and discrete forms, enabling linearization and tractable analysis \cite{Gracy2020, Mieghem2009, Pare2020, Yi2022, Youssef2011}. More recently, Yi et al. \cite{Yi2022} employed mean-field approximations to study discrete SIR models, reformulated them into matrix-based dynamical systems, and studied influence minimization via edge deletion, along with convergence properties of the resulting models.

At present, personalization of user experience has become a dominant trend and a primary objective of online social networks, especially with the increasing use of artificial intelligence algorithms \cite{Silva2016} to enhance user engagement on platforms such as Facebook, Twitter, YouTube and TikTok. Consequently, the study of users' behavior has become particularly relevant in this context \cite{Youssef2011}.

In this paper, we present a novel discrete-time SDIR model, a new extension of the classical discrete SIR framework. The SDIR model incorporates an intermediate state, D (Delayable), between S and I. This new state captures scenarios where a node that receives information does not immediately become “infected,” but may instead exhibit a delay before spreading the information-or may ultimately choose not to spread it at all. This behavioral shade reflects realistic interactions in online social networks, where users exercise discretion in processing, accepting, and sharing information. Each seed node corresponds to a social media account (e.g., a Facebook/Twitter/YouTube/TikTok user account or fan page) that initiates diffusion to its followers at the beginning. In reality, many users either ignore their received information or share it only after a delay in judgment. User procrastination on online social networks has been examined as an inherent behavioral factor in the use of such platforms \cite{Barabasi2005, Qi2018}. By introducing the Delayable state, the SDIR model offers a more accurate representation of individual-level behavior in information diffusion.

Our paper focus on discrete-time SDIR model by using the mean-field approximation method to transform the model SDIR to deterministic SDIR model. Using some techniques in linear algebra and spectral matrix theory, we study optimization problems that minimize the number of infections in SDIR Markov chain model on a network. By using some technique assumptions, we give the sufficient condition for the convergence and the stability of our model. Moreover, we investigate the change of the number of infections after deleting edges. We propose the an efficient modified greedy algorithm based on \cite{Yi2022} and Sandwich algorithm based on Sandwich framework \cite{Lu2015, Wang2017} for minimize infections. 

\section{Proposed Model}

The proposed SDIR model is extended from the SIR model to better suit the application of simulating information diffusion on social networks.
\begin{figure}
	\begin{center}
		\begin{tikzpicture}
			\node[draw, circle, minimum size=0.4cm] at (-2.6,0) (S) {$S_i$};
			\node[draw, circle, minimum size=0.4cm] at (0,-0.6) (I) {$I_i$};
			\node[draw, circle, minimum size=0.4cm] at (2.6,0) (R) {$R_i$};
			\node[draw, circle, minimum size=0.4cm] at (0,0.6) (D) {$D_i$};
			
			\draw[->] (S) -- (I) node[midway, left, below =0.2pt] {$\alpha_i,\beta_{ij}$};
			\draw[->] (I) -- (R) node[midway, left, below=0.2pt] {$\delta_i$};
			\draw[->] (S) -- (D) node[midway, left, above=0.25pt] {$1-\alpha_i$};
			\draw[->] (D) -- (I) node[midway, left] {$\omega_i$};
			\draw[->] (D) -- (R) node[midway, above] {$\delta'_i$};
		\end{tikzpicture}
		\caption{The SDIR model}
		\label{SDIR}
	\end{center}
\end{figure}
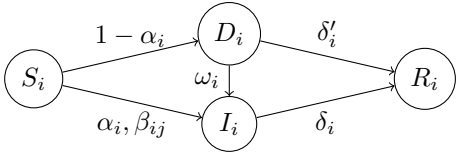
A key characteristic of the SDIR model is the introduction of state D (Delayable) which means that a node can be delayed in its infection process. More specifically, consider a directed graph $G(V,E)$ with $|V|=n$, where each node represents a user and can be in one of four states: \textit{Susceptible} (S), \textit{Delayable} (D), \textit{Infected} (I), or \textit{Recovered} (R). Notably, Delayable means that a delayed state after being exposed to information but not immediately spreading it, and more importantly, it is still considered as an infected state. It is evident that $S_i(t)+I_i(t)+D_i(t)+R_i(t)=1$. At time $t$, a node $i$ in state S can be infected by an adjacent node $j$ in state I that has an edge pointing to $i$ with probability $\beta_{ij}(t)$. If successful, node $i$ will transition to one of two states: state I with probability $\alpha_i(t)$ or state D with probability $1-\alpha_i(t)$. Furthermore, if a node $i$ is in state D at time $t$, it may transition to another state or remain in its current state according to the following rule: A real number $p \in [0,1]$ is randomly chosen following a uniform distribution. Then, if $p<\omega_i(t)$, it transitions to state I; if $\omega_i(t)\leq p < \omega_i(t)+\delta'_i(t)$, it transitions to state R; otherwise, it remains in its current state D. When node $i$ is infected at time $t$, it also heals with rate $\delta_i(t)$. Before establishing the equations for the model, we introduce the following assumptions.
\begin{assumption}\label{as1}
	$\delta_i(t) \leq \delta'_i(t)$, $\forall i\in\overline{1,n},t\geq 0.$
\end{assumption}

\begin{assumption}\label{as2}
	$\alpha_i(t)\in(0,1]$, $\forall i=\overline{1,n},t\geq 0$.
\end{assumption}
The set of inequalities in Assumption \ref{as1} is derived from the observation that during the process of information diffusion and reception, users tend to ignore or quickly forget information (transition to state R) if they do not share or interact immediately.

Next, for every node $i \in {1,2,...,n}$, we have:
\begin{align}
	I_i(t+1) = &S_i(t)\alpha_i(t)\Biggl(1-\prod_{j=1}^n \bigl(1-\beta_{ij}(t)I_j(t)\bigr)\Biggr) \notag\\
	&+ \omega_i(t)D_i(t) + (1-\delta_i(t))I_i(t) \\[0.3em]
	D_i(t+1) = &S_i(t)(1-\alpha_i(t))\Biggl(1-\prod_{j=1}^n \bigl(1-\beta_{ij}(t)I_j(t)\bigr)\Biggr)\notag\\
	&+ (1-\omega_i(t)-\delta'_i(t))D_i(t) \\[0.3em]
	R_i(t+1) = &\delta_i(t)I_i(t)+\delta'_i(t)D_i(t)+R_i(t)
\end{align}
By evaluating $S_i(t) \leq S_i(0)$, taking the expectation on both sides of each equation, and then linearizing the model using the mean-field approximation and the approximation formula $e^x \approx x+1$ when $x\rightarrow 0$ for the quantity $(1-\prod_{j=1}^n(1-\beta_{ij}(t)I_j(t)))$ under the assumption that $\beta_{ij}(t)$ are independent for every pair $i,j$, and the coefficients $\beta_{ij}(t)$, $\delta_i(t)$, $\alpha_i(t)$, $\omega_i(t)$, $\delta'_i(t)$ are independent and identically distributed, we obtain the following system:
\begin{align}
	\bm{x}(t+1) &= (\bm{I}-\bm{D}+\bm{A}\bm{S}(0)\bm{B})\bm{x}(t) \notag \\
	& + \bm{W}\bm{y}(t) \\[0.3em]
	\bm{y}(t+1) &= (\bm{I}-\bm{A})\bm{S}(0)\bm{B}\bm{x}(t) \notag \\
	& + (\bm{I}-\bm{W}-\bm{D'})\bm{y}(t) \\[0.3em]
	\bm{r}(t+1) &= \bm{D}\bm{x}(t)+\bm{D'}\bm{y}(t)+\bm{r}(t)
\end{align}

Here, $\bm{x}(t)$, $\bm{y}(t)$, $\bm{r}(t)$ are vectors whose $i$-th element takes the value $\mathbb{E}[I_i(t)]$, $\mathbb{E}[D_i(t)]$, $\mathbb{E}[R_i(t)]$, respectively. $\bm{A}$, $\bm{W}$, $\bm{D}$, $\bm{D'}$ are diagonal matrices whose $i$-th diagonal element takes the value $\mathbb{E}[\alpha_i(t)]$, $\mathbb{E}[\omega_i(t)]$, $\mathbb{E}[\delta_i(t)]$, $\mathbb{E}[\delta'_i(t)]$, respectively. $\bm{S}(0)$ is a diagonal matrix whose $i$-th diagonal element is $1-x_i(0)-y_i(0)-r_i(0)$. The matrix $\bm{B}$ consists of elements $B_{ij}=\mathbb{E}[\beta_{ij}(t)]$. We also assume that $\sum_{j=1}^nB_{ij}<1$, $\forall i=1,2,\cdots,n$.

\section{Sufficient Condition for Global Convergence}
One difference between the SDIR model and the SEIR model in \cite{Pastor-Satorras2015} is that the delayed state in the SDIR model is still considered an infected state, while the SEIR model only considers E as an exposed state. Therefore, when considering global stability in the SDIR model, we need to find a sufficient condition for $\bm{x}(t)$ and $\bm{y}(t)$ to both approach the zero state. Consider a vector $\bm{q}\in(0,1]^{n}$ which has $n$ elements $q_1$, $q_2$,$\cdots$, $q_n$, then define $$\bm{C}(\bm{q})\coloneqq \text{diag}(\min(D_i,D'_i+(1-1/q_i)W_i)),$$ $$\bm{M}(\bm{q}) \coloneqq \bm{I}-\bm{C}(\bm{q})+(\bm{A}+\bm{Q}(\bm{I}-\bm{A}))\bm{S}(0)\bm{B},$$ where $\bm{Q}=\text{diag}(\bm{q})$. Moreover, if there is no confusion, we convention to write $\bm{M}(\bm{q})$ as $\bm{M}$. Recall that the sufficient condition for the convergence of the SIR model is $\rho(\bm{M}_{\text{SIR}})<1$, where $\bm{M}_{\text{SIR}}=\bm{I}-\bm{D}+\bm{S}(0)\bm{B}$. Under Assumption \ref{as1}, we can give a convergence condition that is better than the SIR model and is stated in the following theorem.

\begin{theorem}\label{convergence}
	Under Assumption \ref{as1}, there always exists a choice of vector $\bm{q}\in(0,1]^{n}$ such that $\rho(\bm{M})\leq\rho(\bm{M}_{\text{SIR}})$. Moreover, if $\rho(\bm{M})<1$ then $\bm{x}(t)$ and $\bm{y}(t)$ converge to the $\bm{0}_{n\times1}$.
\end{theorem}

\begin{corollary}
	If there exist a vector $\bm{q}\in(0,1]^{n}$ such that $\rho(\bm{M})<1$ then $\bm{x}(t)$ and $\bm{y}(t)$ converge to the $\bm{0}_{n\times1}$.
\end{corollary}

Theorem \ref{convergence} provides a "loose" condition for the convergence of the model due to the presence of $n$ parameters $q_i, i=\overline{1,n}$. The variation of these $n$ parameters may change the convergence rate of the model, although it is not certain whether this quantity is positively correlated with the spectral radius of $\bm{M}$ or not. It should also be noted that introducing the delayed state only helps the model to simulate in more detail and better capture the information dissemination behavior of social network users, but it does not assert that the convergence condition of the SDIR model is always better than existing related models. We observe that when Assumption \ref{as1} is removed, the result of Theorem \ref{convergence} may not always be achieved.
\section{The Problems and bounding functions}\label{bounds}
\subsection{The main problem}
Let the vector $\bm{m}(t)\coloneqq \bm{x}(t)+\bm{y}(t)+\bm{r}(t)$ and $\bm{m}^\star\in\mathbb{R}^n$ consisting of elements $m^\star_{i}=\sup\limits_{t}{m_{i}(t)}$.

\begin{definition}
	The quantity $\|\bm{m}^*-\bm{m}(0)\|_1$ is called the number of increased infections on $G$ after the diffusion process ends.
\end{definition}

In the following, we consider the problem of minimizing the diffusion quantity from the nodes in the given seed set by deleting some appropriate edges. Suppose we choose a set of edges $P\subseteq Q$ to delete from the graph. Denote $\bm{B}_{-P}$ as the matrix obtained from $\bm{B}$, $G_{-P}=G(V,E\backslash P)$, and let $\sigma(P)\coloneqq\|\bm{m}^\star-\bm{m}(0)\|_1$ be the function of the number of increased infections in the network when the diffusion ends after deleting the set of edges $P$.

\begin{problem}\label{Problem}
	Given a directed graph $G(V,E)$ with $|V|=n$, representing SDIR diffusion model, an initial state vector $\bm{x}(0)\in[0,1]^n$ and $\bm{y}(0)$, $\bm{r}(0)$ such that $\bm{x}(0)+\bm{y}(0)+\bm{r}(0)\in[0,1]^n$. Let $Q$ be a candidate set of edges such that $Q\subseteq E$ and a positive integer $k$ satisfying $k\leq|Q| = q$. Find a set of edges $P\subseteq Q,|P|\leq k$ to delete from the graph such that the infection amount on $G_{-P}$ is minimized i.e., find
	\[
	P^\star\in \mathop{\operatorname{argmin}}_{P\subseteq Q,|P|\leq k}\sigma(P)
	\]
\end{problem}
An easily noticeable point is that the objective function for Problem \ref{Problem} lacks submodularity or supermodularity. This motivates us to find a good heuristic algorithm that provides a solution with good optimality approximation. Similar to \cite{Yi2022}, if the condition of Theorem \ref{convergence} is satisfied, combined with the constraints of the above assumptions, we can find a monotonic upper bounding function that possesses supermodularity.

\subsection{Supermodular Upper Bound}
\begin{theorem}\label{upbound}
	Under Assumption \ref{as1}, when $\rho(\bm{M}_{-P})<1$, the infection amount of the SDIR model with the removed edge set $P$ does not exceed 
	\begin{align*}
		\sigma^U(P) = &\bm{1}^\top(\bm{A}+\bm{Q}(\bm{I}-\bm{A}))^{-1}\\
		&(\bm{D}(\bm{I}-\bm{M}_{-P})^{-1}-\bm{I})(\bm{x}(0)+\bm{Q}\bm{y}(0)).
	\end{align*}
\end{theorem}
If the constraints of Assumption \ref{as1} are satisfied, it is clear that the obtained upper bound function is not greater than the upper bound function for the D-SIR model of \cite{Yi2022} and is significantly smaller if there exists a node $i$ satisfying $D'_i>D_i$ (the reader is recommended to see the proof of Theorem \ref{convergence} in the appendix for more details). We have the following lemma to prove the monotonicity and supermodularity of the function $\sigma^U(.)$.

\begin{lemma}\label{supup}
	$\sigma^U(.)$ is a non-increasing and supermodular function.
\end{lemma}

\subsection{Supermodular Lower Bound}

A noteworthy point in the SDIR model is that one can find a lower bound function that is also supermodular. Hence, instead of running the greedy algorithm only for the upper bound function $\sigma^U(.)$, one can run the greedy algorithm once more for the following lower bound function $\sigma^L(.)$. This is precisely the idea of the Sandwich approximation and will be presented in more detail in the next section. Let $\bm{N}_{-P} \coloneqq \bm{I}-\bm{D}+\bm{A}\bm{S}(0)\bm{B}_{-P}$.

\begin{theorem}\label{lowbound}
	Under Assumption \ref{as2}, when $\rho(\bm{N}_{-P})<1$ the number of infections of the SDIR model with the set of removed edges $P$ is at least
	\begin{align*}
		\sigma^L(P) = &\bm{1}^\top\bm{A}^{-1}\big(\bm{D}(\bm{I}-\bm{N}_{-P})^{-1}-\bm{I}\big)\\
        &(\bm{x}(0)+\bm{W}(\bm{W}+\bm{D}')^{-1}\bm{y}(0)\big).
	\end{align*}
	
\end{theorem}

\begin{lemma}\label{suplow}
	$\sigma^L(.)$ is a non-increasing and supermodular function. 
\end{lemma}
The proofs of Lemmas \ref{supup} and \ref{suplow} can be argued in a manner similar to that in \cite{Yi2022}. Now, we have $\sigma^L(P)\leq \sigma(P)\leq \sigma^U(P)$,
with $\sigma^L(P),\sigma^U(P)$ being two non-increasing and supermodular functions. Accordingly, the Sandwich approximation principle \cite{Lu2015, Wang2017} implies that the solution returned by the greedy algorithm always guarantees an approximation ratio relative to the optimal solution; specifically,
{\small
\begin{align*}
\sigma^U(\emptyset)-\sigma(P_{\text{Sand}}) \geq&\max{\left\{\frac{\sigma^U(\emptyset)-\sigma(P_{L})}{\sigma^U(\emptyset)-\sigma^L(P_{L})},
\frac{\sigma^U(\emptyset)-\sigma^U(P^\star)}{\sigma^U(\emptyset)-\sigma(P^\star)}\right\}}\\
&\left(1-\frac{1}{e}-\epsilon\right)(\sigma^U(\emptyset)-\sigma(P^\star)),
\end{align*}
}
where $P_{\text{Sand}}=\operatorname*{argmin}_{P\in\{P_{L},P_{0},P_{U}\}}\sigma(P)$ with $P_{L},P_{U}$ being the solutions returned by the greedy algorithm applied to $\sigma^L(.)$ and $\sigma^U(.)$, respectively, $P_{0}$ being a solution returned by some algorithm (possibly greedy) applied to $\sigma(.)$, and $P^\star$ denoting the optimal solution for objective function.
\section{Algorithms}
First, we have the following proposition showing that solving Problem \ref{Problem} still remains computationally hard due to the use of a more general model than SIR model.

\begin{proposition}
	The problem of finding the optimal edge set $P^\star$ to remove in order to minimize the spread in the SDIR network model, as formulated in \ref{Problem}, is NP-hard.
\end{proposition}

Since the problem is NP-hard, it is very difficult to design an algorithm that provides an exact solution in all cases within polynomial time. Fortunately, with the two bound functions derived in the previous section, and when the condition $\rho(\bm{M}_{-P})<1$ is satisfied, we propose the Sandwich approximation algorithm, which is described in Algorithm \ref{sandwich}.

\begin{algorithm}[H]
	\caption{Greedy Algorithm (GA)}
	\label{greedy}
	\textbf{Input:} A function $f\in\left\{\sigma^L,\sigma^U\right\}$, a graph $G$, initial states, a candidate edge set $Q$, and an integer $k$.
	
	\textbf{Output:} An edge set $P\subseteq Q$ of size $k$.
	
	Initialize $P\leftarrow \emptyset$
	
	\For{$i = 1$ \KwTo $n$}{
		Compute $f(P\cup \{ e \})$ for each $e\in Q\backslash P$\\
		$e^\star \leftarrow \operatorname*{argmax}_{e\in Q\backslash P} \big(f(P)-f(P\cup \{ e \})\big)$\\
		$P\leftarrow P\cup \{ e^\star \}$
	}
	
	\Return $P$
\end{algorithm}
Specifically, Algorithm \ref{greedy} is applied to the two functions $\{\sigma^L,\sigma^U\}$ in order to obtain three corresponding edge sets by greedily selecting an edge in each iteration. Based on Lemmas \ref{supup} and \ref{suplow}, the solution returned by this greedy algorithm guarantees an approximation ratio of $(1-1/e-\epsilon)$ for the two functions $\sigma^U(\varnothing)-\sigma^L(.)$ and $\sigma^U(\varnothing)-\sigma^U(.)$ with respect to the optimal solution. Next, among the resulting edge sets, we compute the estimated infection for each set and select the one that yields the smallest spreading influence.
\begin{algorithm}[H]
	\caption{Sandwich Approximation Algorithm (SAA)}
	\label{sandwich}
	\textbf{Input:} A graph $G$, initial states, a candidate edge set $Q$, and an integer $k$.
	
	\textbf{Output:} An edge set $P\subseteq Q$ of size $k$.
	
	$P_L\leftarrow$ Result of GA for $\sigma^L(.)$
	
	$P_0\leftarrow$ Result of a heuristic algorithm for $\sigma(.)$
	
	$P_U\leftarrow$ Result of GA for $\sigma^U(.)$
	
	\Return $P_{\text{Sand}}\leftarrow \operatorname*{argmin}_{P\in\{P_{L},P_{0},P_{U}\}}\sigma(P)$
\end{algorithm}

\begin{remark}
	The solution $P_0$ of $\sigma(.)$ does not play a crucial role in proving the approximation guarantee of SAA. Moreover, the greedy method for $\sigma(.)$ does not provide any guaranteed approximation ratio compared to the optimal solution, while incurring higher computational complexity. Therefore, $P_0$ can be obtained using other simpler methods, such as random selection. We propose two algorithms that are GA and SAA.
\end{remark}
\section{Experiments}
In this section, we present numerical examples to demonstrate the convergence and effectiveness of the proposed algorithms. We evaluate the performance of our method on a real-world dataset and compare it against some known strategies.
\subsection{Comparison with Heuristic Algorithms}
First, we compare the proposed SAA with two heuristic algorithms, namely \textit{Max-degree} \cite{Albert2000} and \textit{Random} \cite{Callaway2000}, applied to the SDIR model. The algorithms are defined as follows:
\begin{enumerate}
    \item \textbf{Max-degree:} This algorithm iteratively removes the edge incident to the node with the highest weighted degree at each step.
    \item \textbf{Random:} This algorithm removes an edge selected uniformly at random at each iteration. Note that the procedure for determining $P_0$ is integrated into this selection process.
\end{enumerate}

We conducted experiments using the real-world contact network dataset collected in Haslemere, England \cite{Klepac2018}. The parameter settings and preprocessing steps follow the methodology described in Yi \emph{et al} \cite{Yi2022}. The new simulation parameters are initialized as follows: The initial state $\bm{y}(0)$ is set to zero for all nodes, except for five seed nodes which are assigned random values in the interval $[0, 0.05]$. Also, the initial state $\bm{x}(0)$ is set to zero for all nodes, except for five seed nodes which are assigned random values in the interval $[0.8, 0.85]$. Next, the infection rates $A_i$ are distributed based on node categories: the five seed nodes have $A_i \in [0.9, 1]$; forty-five randomly selected nodes have $A_i \in [0.65, 0.87]$; and the remaining nodes have $A_i \in [0.15, 0.35]$. For each node $i \in V$, the weight $W_i$ is chosen uniformly at random from $[0.15, 0.35]$. To ensure Assumption \ref{as1} is satisfied, $D'_i$ is selected from the interval $[\max(0, D_i - W_i), 0.95 - W_i]$.

\begin{figure}[htbp]
    \centering
    \includegraphics[width=\linewidth]{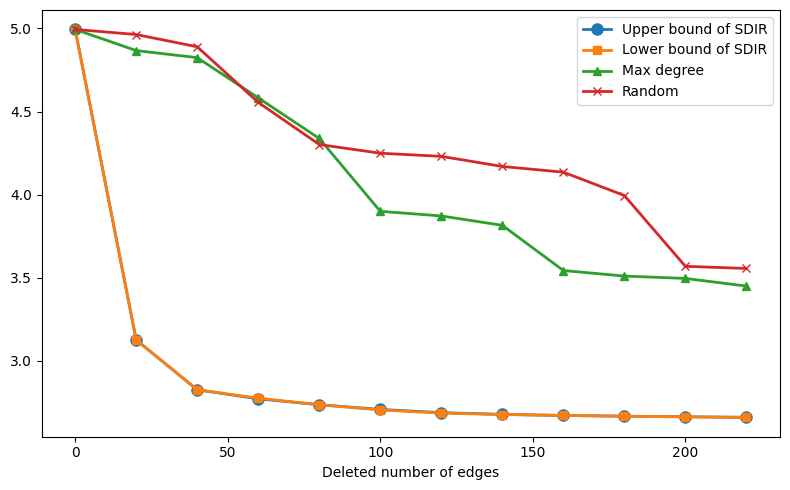} 
    \caption{Haslemere Network with preprocessing}
    \label{BBC}
\end{figure}

The simulation results are presented in Fig.~\ref{BBC}. It is evident that the solutions obtained from the proposed sandwich bounding functions significantly outperform the two heuristic algorithms. The number of increased infections decreases rapidly within the removal of the first 50 edges. Notably, the upper and lower bounds derived from the sandwich algorithms are very close to each other, indicating the tightness of the approximation with respect to the objective function.
It is similar to the results or the complexity of the Greedy algorithm in \cite{Yi2022}, in our simulations, the running time of our Algorithm \ref{greedy} for the SDIR model is $O(n^{3} + k(n^{2} + qn))$.

\subsection{Convergence Analysis}
\begin{figure}[H]
    \centering
    \includegraphics[width=\linewidth]{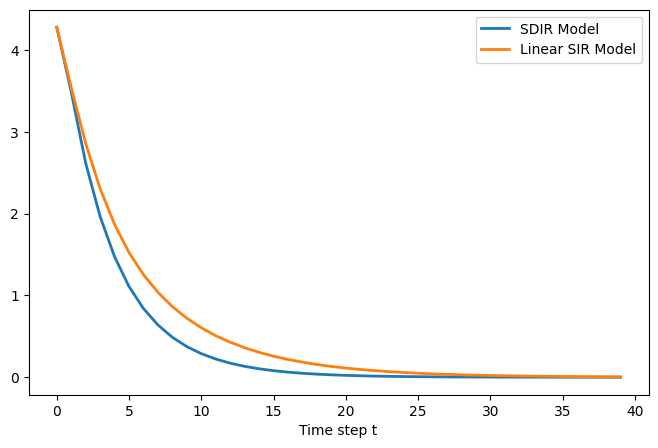}
    \caption{Convergence of infection states over time}
    \label{Converge}
\end{figure}
Next, utilizing the same parameter setup, we examine the convergence of the infection states in both the proposed SDIR model and the linear SIR model. The total amount of infection state of the SDIR model is defined as $||\bm{x}(t) + \bm{y}(t)||_1$. For comparison, the total amount of infection state of the linear SIR model is calculated as $||\hat{\bm{x}}(t)||_1$, where the initial condition is set to $\hat{\bm{x}}(0) = \bm{x}(0) + \bm{y}(0)$.

The results illustrated in Fig.~\ref{Converge} demonstrate that the SDIR model exhibits faster convergence compared to the linear SIR model. Specifically, the infection states in the SDIR model rapidly decays to $\bm{0}_{n\times 1}$ within the first 10 time steps. In contrast, the linear SIR model exhibits a slower decay, requiring approximately 5 additional time steps to reach a similar extinction state. This empirical observation reinforces the theoretical validity of Theorem~\ref{convergence} under Assumption~\ref{as1}.
\section{Conclusion}
We studied the discrete-time SDIR model, a new modified SIR model. We showed that the sufficient condition for the convergence and the stability of the SDIR model. Furthermore, we investigate the problem that minimize the number of infections where deleting some edges. Two algorithms, greedy and Sandwich algorithms were proposed to solve the minimizing infections problem. Some experimental results were given.

\section*{Appendix}

\subsection{Proof of Theorem \ref{convergence}}
Under Assumption \ref{as1}, it is clear that $D_i\leq D_i'$, $\forall i=1,2,\ldots,n$. Hence we can choose $q_i\in\left[\dfrac{W_i}{W_i+D_i'-D_i},1\right]$. We have: $\bm{x}(t+1)+\bm{Q}\bm{y}(t+1)$
\begin{align*}
&= (\bm{I}-\bm{D}+(\bm{A}+\bm{Q}(\bm{I}-\bm{A}))\bm{S}(0)\bm{B})\bm{x}(t) \\
&\qquad +(\bm{Q}\bm{I}-\bm{Q}\bm{D'}-(\bm{I}-\bm{Q})\bm{W})\bm{y}(t)\\[6pt]
&= (\bm{I}-\bm{D}+(\bm{A}+\bm{Q}(\bm{I}-\bm{A}))\bm{S}(0)\bm{B})\bm{x}(t) \\
&\qquad +(\bm{I}-\bm{D'}-(\bm{Q}^{-1}-\bm{I})\bm{W})\bm{Q}\bm{y}(t)\\[6pt]
&\leq (\bm{I}-\bm{C}+(\bm{A}+\bm{Q}(\bm{I}-\bm{A}))\bm{S}(0)\bm{B})\bm{x}(t) \\
&\qquad +(\bm{I}-\bm{C}+(\bm{A}+\bm{Q}(\bm{I}-\bm{A}))\bm{S}(0)\bm{B})\bm{Q}\bm{y}(t)\\[6pt]
&= \bm{M}(\bm{x}(t)+\bm{Q}\bm{y}(t)),\forall t\geq 0.
\end{align*}

Therefore, $
\bm{x}(t)+\bm{Q}\bm{y}(t) \leq \bm{M}^t(\bm{x}(0)+\bm{Q}\bm{y}(0))$. If $\rho(\bm{M}) \le 1-\epsilon$, then $\|\bm{x}(t)+\bm{Q}\bm{y}(t)\|\leq (1-\epsilon)^t\|\bm{x}(0)+\bm{Q}\bm{y}(0)\|$,
which forces $\bm{x}(t)$ and $\bm{y}(t)$ to converge to $0$ as $t\to\infty$. This completes the proof.
\subsection{Proof of Theorem \ref{upbound}}
From the proof of Theorem \ref{convergence}, we have $\bm{x}(t) \leq \bm{M}_{-P}^t(\bm{x}(0)+\bm{Q}\bm{y}(0))$, $\forall t\geq 0$. Next, observe that
\begin{align*}
&\bm{m}(t)
= \bm{x}(t)+\bm{y}(t)+\bm{r}(t) \\
&= \bm{x}(0)+\bm{y}(0)+\bm{r}(0)+\bm{S}(0)\bm{B}_{-P}\sum_{l=0}^{t-1}\bm{x}(l)\\
&\leq \bm{m}(0)+\bm{S}(0)\bm{B}_{-P}\sum_{l=0}^{t-1}\bm{M}_{-P}^{l}\big(\bm{x}(0)+\bm{Q}\bm{y}(0)\big).\\
&\Rightarrow \|\bm{m}^\star-\bm{m}(0)\|_1
= \lim_{t\rightarrow\infty}\|\bm{m}(t)-\bm{m}(0)\|_1\\
&\leq \bm{1}^\top\bm{S}(0)\bm{B}_{-P}(\bm{I}-\bm{M}_{-P})^{-1}(\bm{x}(0)+\bm{Q}\bm{y}(0)).
\end{align*}

This completes the proof.

\subsection{Proof of Theorem \ref{lowbound}}

First, we have the following lemma:
\begin{lemma}\label{splemma}
    Let $\bm{A}$ and $\bm{B}$ be two square nonnegative real matrices of the same dimension with spectral radius $\rho(\bm{A})<1$ and $\rho(\bm{B})<1$. Then
    \[
L=\lim_{t\rightarrow\infty}\sum_{s=0}^{t}\sum_{l=0}^{s}\bm{A}^{l}\bm{B}^{\,s-l} = (\bm{I}-\bm{A})^{-1}(\bm{I}-\bm{B})^{-1}.
    \]
\end{lemma}
\begin{proof}
Denote $\bm{C}_{l}=\bm{I}+\bm{B}+\cdots+\bm{B}^{l}$ for $l=0,1,2,\ldots$. For $t>t_0$ with fixed $t_0\in\mathbb{N}^\star$ we have $\sum_{s=0}^{t}\sum_{l=0}^{s}\bm{A}^{l}\bm{B}^{\,s-l}
\geq \sum_{l=0}^{t_{0}}\bm{A}^l\bm{C}_{t-t_{0}} + \bm{A}^{t_{0}+1}\sum_{l=0}^{t-t_0-1}\bm{A}^{l}$.
Fixing $t_0$ and letting $t\rightarrow\infty$ yields $L\geq (\bm{I}+\cdots+\bm{A}^{t_{0}})(\bm{I}-\bm{B})^{-1}+\bm{A}^{t_{0}+1}(\bm{I}-\bm{A})^{-1}$.
Now letting $t_0\rightarrow\infty$ gives $L\geq (\bm{I}-\bm{A})^{-1}(\bm{I}-\bm{B})^{-1}$.
Similarly, one may show the reverse bound by observing $\sum_{s=0}^{t}\sum_{l=0}^{s}\bm{A}^{l}\bm{B}^{\,s-l}
\leq \sum_{l=0}^{t_{0}}\bm{A}^{l}\bm{C}_{t} + \bm{A}^{t_{0}+1}\sum_{l=0}^{t-t_0-1}\bm{A}^{l}\bm{C}_{t-t_{0}-1}$,
and repeating the same limiting argument to obtain $L\leq (\bm{I}-\bm{A})^{-1}(\bm{I}-\bm{B})^{-1}$. Thus the lemma is proved. 
\end{proof}

Returning to Theorem \ref{lowbound}. We have $\bm{x}(t+1)=\bm{N}_{-P}\bm{x}(t)+\bm{W}\bm{y}(t)=\bm{N}_{-P}^{\,t+1}\bm{x}(0)+\sum_{l=0}^{t}\bm{N}_{-P}^{\,t-l}\bm{W}\bm{y}(l), \forall t\geq 0$. It is easy to see that $\bm{y}(t)\geq \bm{F}\bm{y}(t-1)$ with $\bm{F}=\bm{I}-\bm{W}-\bm{D}'$, it follows that $\bm{y}(t)\geq \bm{F}^{t}\bm{y}(0)$ for all $t\geq 0$. So
\begin{align*}
&\bm{m}(t+1)-\bm{m}(0)
= \bm{S}(0)\bm{B}_{-P}\Bigg(\sum_{s=1}^{t}\Big(\bm{N}_{-P}^{s}\bm{x}(0)
\\
&+ \sum_{l=0}^{s-1}\bm{N}_{-P}^{\,s-1-l}\bm{W}\bm{y}(l)\Big)+\bm{x}(0)\Bigg)\geq\\
&\bm{S}(0)\bm{B}_{-P}\sum_{l=0}^{t}\bm{N}_{-P}^{\,l}\bm{x}(0)+ \bm{S}(0)\bm{B}_{-P}\sum_{s=0}^{t-1}\sum_{l=0}^{s}\bm{N}_{-P}^{\,s-l}\bm{F}^{\,l}\bm{W}\bm{y}(0).
\end{align*} Setting $t\rightarrow\infty$ and applying Lemma \ref{splemma} yields $\bm{m}^\star-\bm{m}(0)
\geq \bm{S}(0)\bm{B}_{-P}\Big((\bm{I}-\bm{N}_{-P})^{-1}\bm{x}(0)
+(\bm{I}-\bm{N}_{-P})^{-1}(\bm{I}-\bm{F})^{-1}\bm{W}\bm{y}(0)\Big)$. Then, equivalently, $\|\bm{m}^\star-\bm{m}(0)\|_1
\geq \bm{1}^\top\bm{A}^{-1}\big(\bm{D}(\bm{I}-\bm{N}_{-P})^{-1}-\bm{I}\big)
\big(\bm{x}(0)+\bm{W}(\bm{W}+\bm{D}')^{-1}\bm{y}(0)\big)$.

The theorem is proved.

\begin{thebibliography}{00}
	\bibitem{Acemoglu2013}
	D. Acemoglu, A. Malekian, and A. Ozdaglar, ``Network security and contagion," J. of Economic Theory, 166, 2013, pp. 536--585.
	
	\bibitem{Albert2000}
	R. Albert, H. Jeong, and A.-L. Barabási,``Error and attack tolerance of complex networks,” Nature, 406.6794, 2000, pp. 378--382.
	
	\bibitem{Alpcan2010}
	T. Alpcan and T. Basar, Network security: A decision and game-theoretic	approach, Cambridge University Press, 2010.
	
	\bibitem{Bailey1975}
	N. T. Bailey, The mathematical theory of infectious diseases and its applications, Griffin, London, 1975.
	
	\bibitem{Ahn2013}
	H. J. Ahn, B. Hassibi, ``Global dynamics of epidemic spread over complex networks," Proceedings of the 52nd IEEE Conference on Decision and Control (CDC), IEEE, 2013, pp. 4579--4585.

    \bibitem{Barabasi2005}
    A. L. Barabasi, ``The origin of bursts and heavy tails in human dynamics," Nature, 435, 2005, pp. 207--211.
	
	\bibitem{Callaway2000}
	D. S. Callaway, M. E. J. Newman, S. H. Strogatz, D. J. Watts, ``Network robustness and fragility: Percolation on random graphs,” Physical Review letters, 85, 2000, pp. 5468--5471.
	
	\bibitem{Chen2014}
	W. Chen, C. Castillo, L. V. Lakshmanan, Information and influence propagation in social networks, Morgan \& Claypool Publishers, 2014. 
	
	\bibitem{Jacquet2010}
	P. Jacquet, B. Mans, and G. Rodolakis, ``Information propagation speed in mobile and delay tolerant networks," IEEE Transactions on Information Theory, 56 (10), 2010, pp. 5001--5015.
	
	\bibitem{Gracy2020}
	S. Gracy, P. E. Pare, H. Sandberg, K. H. Johansson, ``Analysis and distributed control of periodic epidemic processes," IEEE Trans. Control Netw. Syst., 8, 2021, pp. 123--134.
	
	\bibitem{Kempe2003}
	D. Kempe, J. M. Kleinberg, and E. Tardos, ``Maximizing the spread of influence through a social network," Proceedings of the 9th ACM SIGKDD International Conference on Knowledge Discovery and Data Mining (KDD), 2003, pp. 137--146.
	
	\bibitem{Kempe2015}
	D. Kempe, J. Kleinberg, E. Tardos, ``Maximizing the Spread of Influence through a Social Network," Theory of Computing, 11 (4), 2015, pp. 105--147.
	
	\bibitem{Kimura2008}
	M. Kimura, K. Saito, and H. Motoda, ``Solving the Contamination Minimization Problem on Networks for the Linear Threshold Model," PRICAI (Lecture Notes in Computer Science), 5351, 2008, pp. 977--984.
	
	\bibitem{Klepac2018}
	P. Klepac, S. Kissler, and J. Gog, ``Contagion! the bbc four pandemic–the model behind the documentary." Epidemics, 24, 2018, pp. 49--59.
	
	\bibitem{Lu2015}
	W. Lu, W. Chen, L. V. Lakshmanan, ``From Competition to Complementarity: Comparative Influence Diffusion and Maximization," Proc.
	VLDB Endow., 9 (2), 2015, pp. 60--71.
	
	\bibitem{Mieghem2009}
	P. Van Mieghem, ``Virus Spread in Networks," IEEE/ACM Transactions on Networking, 17 (1), 2009, pp. 1--14.
	
	\bibitem{Mieghem2011}
	P. Van Mieghem, Graph spectra for complex networks, Cambridge University Press, Cambridge, 2011.
	
	\bibitem{Mieghem2014}
	P. Van Mieghem, F. D. Sahnehz, C. Scoglioz, ``An upper bound for the epidemic threshold in exact Markovian SIR and SIS epidemics on networks," 53rd IEEE Conference on Decision and Control, Los Angeles, CA, USA, 2014, pp. 6228--6233.
	
	\bibitem{Nowzari2016}
	C. Nowzari, V. M. Preciado, G. J. Pappas, ``Analysis and control of epidemics: A survey of spreading processes on complex networks," IEEE Control Syst., 36, 2016, pp. 26--46.
	
	\bibitem{Pare2020}
	P. E. Pare, J. Liu, C. Beck, B. Kirwan, T. Basar, ``Analysis, Estimation, and Validation of Discrete-Time Epidemic Processes," IEEE Transactions on Control Syst. Tech., 28 (1), 2020, pp. 79--93. 
	
	\bibitem{Pastor-Satorras2015}
	R. Pastor-Satorras, C. Castellano, P. Van Mieghem, A. Vespignani, ``Epidemic processes in complex networks," Rev. Mod. Phys. 87, 2015, pp. 925--979.
	
	\bibitem{Canh2018}
	Pham, C.V., Thai, M.T., Duong, H.V., Bao. Q.B., Hoang. X.H., ``Maximizing misinformation restriction within time and budget constraints", J. Comb. Optim. 35, 2018, pp. 1202--1240.
	
	\bibitem{Canh2019}
	Pham, C.V., Phu, Q.V., Hoang, H.X., J. Pey, My. T. Thai, ``Minimum budget for misinformation blocking in online social networks," J. Comb. Optim. 38, 2019, pp. 1101--1127.
	
	\bibitem{Phelps2004}
	J. E. Phelps, R. Lewis, L. Mobilio, D. Perry, and N. Raman, ``Viral marketing	or electronic word-of-mouth advertising: Examining consumer responses and motivations to pass along email," Journal of Advertising research, 44 (04), 2004, pp. 333--384.

    \bibitem{Qi2018}
    J. Qi, X. Liang, Y. Wang, H. Cheng, ``Discrete time information diffusion in online social networks: micro and macro perspectives," Sci. Rep. 8, 2018, pp. 11872.
    
	\bibitem{Richardson2002}
	M. Richardson and P. Domingos, ``Mining knowledge-sharing sites for viral marketing," Proceedings of the 8th ACM SIGKDD international conference on Knowledge discovery and data mining, 2002, pp. 61--70.
	
	\bibitem{Ruhi2015}
	A. Ruhi, B. Hassibi, ``SIRS epidemics on complex networks: Concurrence of exact Markov chain and approximated models," Proceedings of the 54th IEEE Conference on Decision and Control (CDC), IEEE, 2015, pp. 2919--2926.
	
	\bibitem{Kermack1927}
	W. O. Kermack and A. G. McKendrick, ``A contribution to the mathematical theory of epidemics," Proceedings of the Royal Society of London A: Mathematical, Physical and Engineering Sciences, 115 (772), 1927, pp. 700–721.
	
	\bibitem{Shakarian2015}
	P. Shakarian, A. Bhatnagar, A. Aleali, E. Shaabani, R. Guo, Diffusion in Social Networks, Springer, 2015, xi+101 pp.
	
	\bibitem{Silva2016}
	T. C. Silva; L. Zhao, Machine learning in complex networks, Springer, Cham, 2016. xviii+331 pp.
	
	\bibitem{Xie2023}
	J. Xie, F. Zhang, K. Wang, X. Lin, and W. Zhang, ``Minimizing the Influence of Misinformation via Vertex Blocking," ICDE. IEEE, 2023, pp. 789--801.
	
	\bibitem{Wang2003}
	Y. Wang, D. Chakrabarti, C. Wang, and C. Faloutsos, ``Epidemic spreading in real networks: An eigenvalue viewpoint," Proc. 22nd Int. Symp. Reliable Distributed Systems (SRDS’03), 2003, pp. 25--34.
	
	\bibitem{Wang2017}
	Z. Wang, Y. Yang, J. Pei, L. Chu, E. Chen, ``Activity maximization by effective information diffusion in social net-works," IEEE Transactions on Knowledge and Data Engineering, 29.11 2017, pp. 2374–2387.
	
	\bibitem{Yi2022}
	Y. Yi, L. Shan, P. Pare, K. H. Johansson, ``Edge deletion algorithms for minimizing spread in SIR epidemic models," {SIAM Journal on Control and Optimization}, 60.2, 2022, pp. 246--273.
	
	\bibitem{Youssef2011}
	M. Youssef, C. Scoglio, ``An individual-based approach to SIR epidemics in contact networks," J. Theoret. Biol., 283, 2011, pp. 136-144.
\end{thebibliography}
\end{document}